\documentclass[conference]{IEEEtran}
\usepackage{cite}

\ifCLASSINFOpdf
   \usepackage[pdftex]{graphicx}
   \graphicspath{{./Graphics/}}
   \DeclareGraphicsExtensions{.pdf,.jpeg,.png}
\else
\fi
\usepackage{pgf}
\usepackage{tikz}
\usetikzlibrary{arrows,automata}
\usepackage[latin1]{inputenc}

\usepackage{amsmath,amsfonts,amsthm} 
\usepackage{amssymb}
\newtheorem{thm}{Theorem}
\newtheorem{corollary}{Corollary}

\newtheorem{lemma}{Lemma}
\ifCLASSOPTIONcompsoc
  \usepackage[caption=false,font=normalsize,labelfont=sf,textfont=sf]{subfig}
\else
  \usepackage[caption=false,font=footnotesize]{subfig}
\fi

\newcommand{\E}{\mathbb{E}}
\newcommand{\Prob}{\mathbb{P}}
\newcommand{\nn}{\nonumber\\}
\newcommand{\lp}{\left(}
\newcommand{\rp}{\right)}

\begin{document}
%
\title{Status Updates in a multi-stream M/G/1/1 preemptive queue}

\author{\IEEEauthorblockN{Elie Najm}
\IEEEauthorblockA{LTHI, EPFL, Lausanne, Switzerland\\
Email: elie.najm@epfl.ch}
\and 
\IEEEauthorblockN{Emre Telatar}
\IEEEauthorblockA{LTHI, EPFL, Lausanne, Switzerland\\
Email: emre.telatar@epfl.ch}
}


%


\maketitle

\begin{abstract}
We consider a source that collects a multiplicity of streams of updates and sends them through a network to a monitor. However, only a single update can be in the system at a time.
Therefore, the transmitter always preempts the packet being served when a new update is generated. We consider Poisson arrivals for
each stream and a common general service time, and refer to this system as the multi-stream M/G/1/1 queue with preemption.
Using the detour flow graph method, we compute a closed form expression for the average age and the average peak age of each stream.
Moreover, we deduce that although all streams are treated equally from a transmission point of view (they all preempt each
other), 
one can still prioritize a stream from an age point of view by simply increasing its generation rate. However, this
will increase the sum of the ages which is minimized when all streams have the same update rate.  
\end{abstract}


%
\IEEEpeerreviewmaketitle

\section{Introduction}
\label{sect:sect_intro}

Previous work on status update, e.g. \cite{KaulYatesGruteser-2012Infocom,2012CISS-KaulYatesGruteser,CostaCodreanuEphremides2016,
KamKompellaEphremides2013ISIT,NajmNasser-ISIT2016,YatesKaul-2012ISIT}, used an Age of Information (AoI) metric in order to
assess the freshness of randomly generated updates sent by one or multiple sources to a monitor through the network. In these
papers, updates are assumed to be generated according to a Poisson process and the main metric used to quantify the \emph{age}
is the time average age (which we will call average age) given by
\begin{equation}
\label{eq:eq_age_definition}
\Delta = \lim_{\tau\to\infty} \frac{1}{\tau}\int_0^\tau \Delta(t)\mathrm{d}t,
\end{equation} 
where $\Delta(t)$ is the instantaneous age at the receiver of the information about the source status. If the last successfully received update was generated at time
$u(t)$ then the \emph{age} of the source status at time $t$ is $\Delta(t)=t-u(t)$. When the system is idle or an update is being transmitted then
the instantaneous age increases linearly with time. Once an update generated at time $t_i$ is
received by the monitor at $t'_i$, $\Delta(t)$ drops to the value $t'_i-t_i$. This results in the sawtooth sample path seen in
Fig.~\ref{fig1}. 

Moreover, in \cite{CostaCodreanuEphremides2014ISIT} the authors introduce another age metric: the \emph{average peak age}
defined as the time average of the maximum value of the instantaneous age $\Delta(t)$ right before the reception by the monitor of a new update. In
Fig.~\ref{fig1} the peak age right before the reception of the $j^{th}$ successfully transmitted update is denoted by $K_j$.
Hence the average peak age is given by
\begin{equation}
	\label{eq:eq_avg_peak_age_def}
	\Delta_{peak} = \lim_{N\to\infty}\frac{1}{N}\sum_{j=1}^N K_j.
\end{equation}
In this paper, we assume that an \lq observer\rq\ (which we will call source), generating updates according to a Poisson process with
rate $\lambda$, observes $M$ streams of data. At each generation instant, the source chooses to \lq observe\rq\ stream $i$
and send its observation (update) of this stream with probability $p_i$, $i=1,\dots,M$. This probability distribution is a
design parameter that one can control.
Moreover, we assume that the system can handle only one update at a time without any buffer to store incoming updates. This means that whenever a
new update is generated and the system is busy, the transmitter preempts the packet being served and starts sending the new
update instead. Since we consider a general service time distribution for the updates, we denote this transmission scheme by
M/G/1/1 preemptive queue. It has been shown that for a single-stream source and  exponential update service time, preemption
ensures the lowest average age \cite{2012CISS-KaulYatesGruteser}. However, the work in \cite{NajmNasser-ISIT2016} suggests that
under the assumption of gamma distributed service time, preemption might not be the best policy. In
\cite{NajmYatesSoljaninMG11}, the authors derive a closed form expression for the average age of a single-stream source and M/G/1/1
preemptive queue.

As a generalization of the result in \cite{NajmYatesSoljaninMG11}, we derive in this paper a closed form expression for the
average age and average peak age per stream of the multi-stream source M/G/1/1 preemptive queue. To that end we use the detour
flow graph method which is also used to find an upper bound on the error probability of a Viterbi decoder
(see \cite{rimoldi_2016}). A special case of this problem was studied in
\cite{YatesKaul-2016arxiv} where the service time is assumed to be exponentially distributed. In this paper the average age of each stream was obtained in closed form using a stochastic hybrid system.
Another related work, \cite{HuangModiano2015ISIT}, gives closed form expressions
for the average peak age of multi-stream source M/G/1 queues as well as M/G/1/1 queues with blocking. In this last model, if a
newly generated update finds the system busy, it is discarded.

In addition, given a fixed total update rate $\lambda$, we show in this work that if we want to decrease the age of a
certain stream $i$ with respect to other streams we need to increase its update rate (by increasing its choice probability
$p_i$) and thus decreasing the update rates of the other streams. Moreover, if we choose the sum of the ages as our performance
metric and we wish to minimize it then we prove that we need to adopt a fair strategy: all streams should be given the same update
rate.

This paper is structured as follows: in Section~\ref{sec:sec_system_model}, we start by defining the model and the different
variables needed in our study. In Section~\ref{sec:sec_age_multi_stream} we derive the closed form expressions of the
average age and average peak age and state the conditions necessary to minimize the sum of the ages.
\section{System Model}
\label{sec:sec_system_model}
In this model a source generates updates according to a Poisson process with rate $\lambda$ and send them through the network.
However, we assume that the updates belong to $M$ different streams, each stream $i$ being chosen independently at generation time 
with probability $p_i$, $\sum_{i=1}^M p_i =1$. This setup is equivalent to having $M$ independent Poisson sources with rates $\lambda_i = \lambda p_i$,
$i=1,\dots,M$, and $\lambda = \lambda_1+\dots+\lambda_M$ (see \cite{ross}). Moreover, we consider an M/G/1/1 queue with
preemption. This means that only one update can be in the system at a time and thus the different streams preempt each others
and even the same stream preempts itself. This setup was analyzed in \cite{YatesKaul-2016arxiv} where the authors considered an
exponential service time. In this paper, we assume a service time $S$ with general distribution. Given that the system is
symmetric from the point of view of each stream, we will focus --- without loss of generality--- on stream $1$ as the main
stream. Hence, unless stated otherwise, all random variables correspond to packets from stream $1$.

Moreover, in this paper we follow the convention where a random variable $U$ with no subscript corresponds to the
steady-state version of $U_j$ which refers to the random variable relative to the $j^{th}$ received packet from stream $1$. To differentiate
between streams we will use superscripts, so $U^{(i)}$ corresponds to the steady-state variable $U$ relative to the $i^{th}$
stream. 

It is important to note that in M/G/1/1 queues with preemption, some updates might be dropped. Hence we call the updates that 
are not dropped, and thus delivered to the receiver, as \lq\lq successfully received updates\rq\rq\ or \lq\lq successful updates\rq\rq. 
We also define: $(i)$ $Y_j= t'_{j+1}-t'_{j}$ to be the interdeparture time between the $j^{th}$ and ${j+1}^{th}$ 
successfully received updates, $(ii)$ $X^{(i)}$ to be the interarrival time between two consecutive generated updates from stream
$i$, $i=1,\dots,M$, (which may or may not be successfully transmitted), so $f_{X^{(i)}}(x) = \lambda_i e^{-\lambda_i x}$, $(iii)$ $S$ to be the
service time random variable for any update (from any stream) with distribution $F_S(t)$, $(iv)$ $T_j$ to be the system time, or the time spent
by the $j^{th}$ successful update in the queue and $(v)$ $N_\tau=\max\left\{n:t'_{n}\leq\tau\right\}$, the number of 
successfully received updates from stream $1$ in the interval $[0,\tau]$. In our model, we assume the service time of the 
updates from the different streams to be independent of the interarrival time between consecutive packets (belonging to the same
stream or not). These concepts are illustrated in Fig.~\ref{fig1}, where only successfully transmitted packets from stream $1$
are shown.

\begin{figure}[!t]
	\centering
	\includegraphics[scale=0.5]{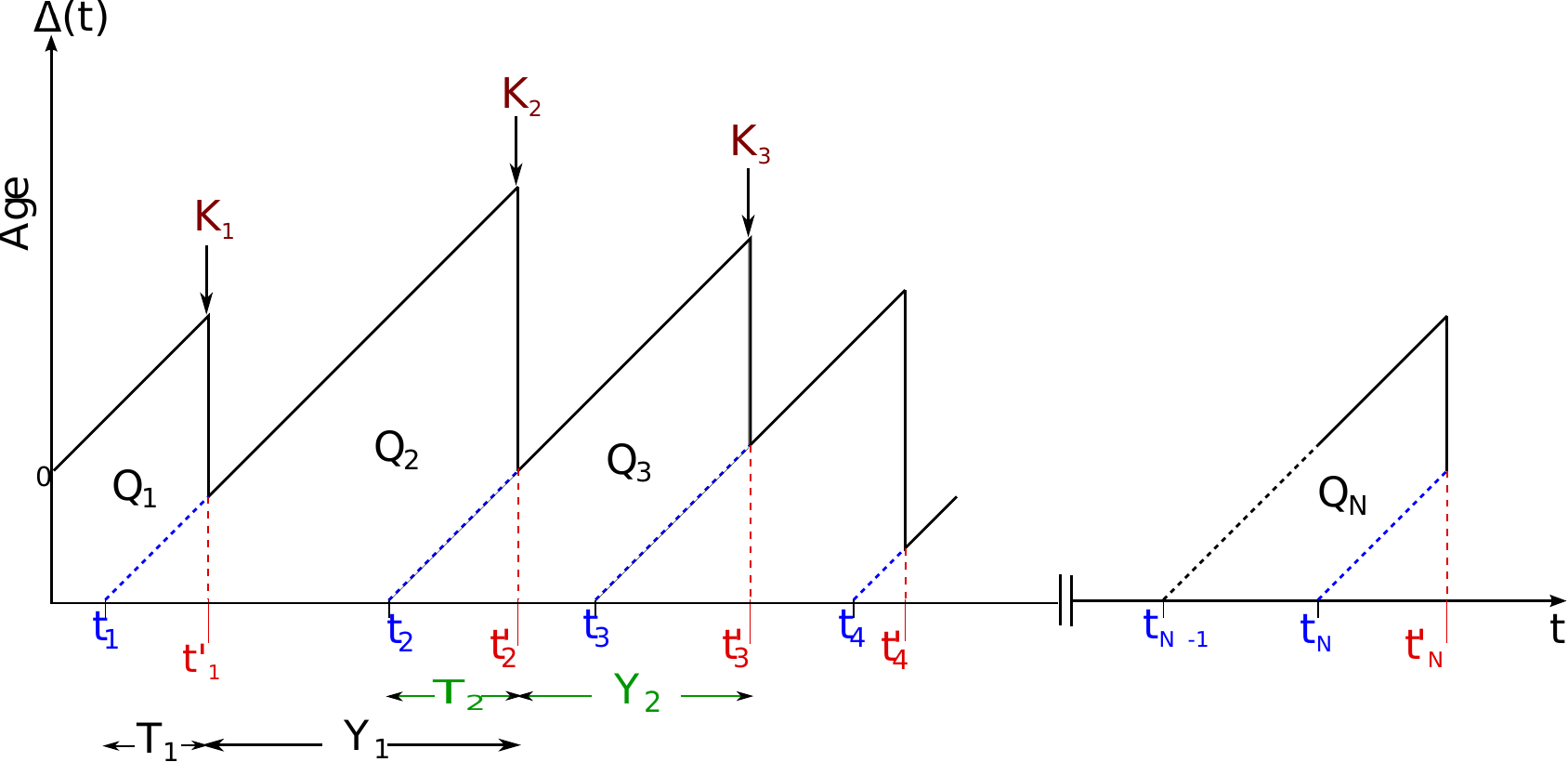}%
	\caption{Variation of the instantaneous age of stream $1$ for M/G/1/1 queue with preemption}
	\label{fig1}
\end{figure}

\section{Age of a multi-stream M/G/1/1 preemptive queue}
\label{sec:sec_age_multi_stream}
We denote by $P_\lambda$, the Laplace transform of the service time distribution evaluated at
$\lambda=\lambda_1+\dots+\lambda_M$, i.e.
$P_\lambda = \E\lp e^{-\lambda S}\rp$.

Before stating the main result of this section we need the following lemmas.
\begin{lemma}
	\label{lemma:lemma_general_results}
	Let $X$, $\Lambda$ and $S$ be three non-negative independent random variables with respective distributions: $f_X(x) = \lambda_1
	e^{-\lambda_1 x}$, $f_\Lambda(x) = (\lambda-\lambda_1)e^{-(\lambda-\lambda_1)x}$  and $f_S(t)$, with
	$\lambda>\lambda_1>0$. Let $A$, $Z$, $B$, $V$ be random
	variables such that $\Prob\lp A>t\rp=\Prob\lp X>t|X<\Lambda\rp$, $\Prob\lp Z>t\rp=\Prob\lp \Lambda>t|X>\Lambda\rp$,
	$\Prob\lp B>t\rp=\Prob\lp X>t|X<\min\lp S,\Lambda\rp\rp$ and $\Prob\lp V>t\rp=\Prob\lp \Lambda>t|\Lambda<\min\lp
	S,X\rp\rp$. Then,
	\begin{description}
		\item[$(i)$] $\E\lp e^{sA}\rp = \E\lp e^{sZ}\rp = \frac{\lambda}{\lambda-s}$,
		\item[$(ii)$] $\E\lp e^{sB}\rp = \E\lp e^{sV}\rp =
			\frac{\lambda\lp1-P_{\lambda-s}\rp}{\lp\lambda-s\rp\lp1-P_\lambda\rp}$,
	\end{description}
	with $P_\lambda$ being the Laplace transform of the random variable $S$ evaluated at $\lambda$.
\end{lemma}

\begin{proof}
	We will only prove the result for the variable $B$ since we can apply the same technique for the others.
	Denote by $\bar{F_S}(t)$ the complementary CDF of $S$. Then, 
	\begin{align*}
		\Prob\lp \min(S,\Lambda)\geq t\rp &=\Prob\lp S\geq t,
	\Lambda\geq t\rp\\&= \Prob\lp S\geq t\rp\Prob\lp \Lambda\geq t\rp\\ 
						  &= \bar{F_S}(t)e^{-(\lambda-\lambda_1)t}.
	\end{align*}
	{\small\begin{align*}
		f_B(t) &= \lim_{\epsilon\to 0}\frac{\Prob\lp B\in [t,t+\epsilon] \rp}{\epsilon}\\
		       &= \lim_{\epsilon\to 0}\frac{\Prob\lp X\in [t,t+\epsilon]|X\leq\min(S,\Lambda)\rp}{\epsilon}\\
	      	       &= \lim_{\epsilon\to 0}\frac{\Prob\lp X\in [t,t+\epsilon]\rp \Prob\lp X\leq\min(S,\Lambda)| X\in
		       [t,t+\epsilon]\rp}{\epsilon\Prob\lp X\leq\min(S,\Lambda)\rp}\\
		       &= \frac{\lambda_1 e^{-\lambda_1 t}\Prob\lp \min(S,\Lambda)\geq t\rp}{\Prob\lp X\leq\min(S,\Lambda)\rp} =
		       \frac{\lambda_1e^{-\lambda t}\bar{F_S}(t)}{\Prob\lp X\leq\min(S,\Lambda)\rp},
	       \end{align*}}
	\begin{align*}
		\Prob\lp X\leq\min(S,\Lambda)\rp &= \int_0^\infty \Prob\lp \min(S,\Lambda)\geq t|X=t\rp\lambda_1 e^{-\lambda_1
		t}\mathrm{d}t\\
					      	 &= \int_0^\infty \lambda_1 e^{-\lambda t}\bar{F_S}(t)\mathrm{d}t
						 = \frac{\lambda_1}{\lambda}\lp1-P_\lambda\rp,
	\end{align*}
	where the last equality is obtained using integration by parts. Thus $f_B(t) = \frac{\lambda e^{-\lambda
	t}\bar{F_S}(t)}{1-P_{\lambda}}$. Using again integration by parts we find that $
		\E\lp e^{sB}\rp=\int_0^\infty
	f_B(t)e^{st}\mathrm{d}t = \frac{\lambda\lp1-P_{\lambda-s}\rp}{\lp\lambda-s\rp\lp1-P_\lambda\rp}$.
\end{proof}
\begin{lemma}
	\label{lemma:lemma_mg11_sys_time}
	For the M/G/1/1 queue with preemption described above, the moment generating function of the system time $T^{(i)}$
	corresponding to a stream $i$ is given by
	\begin{equation}
		\label{eqn:eqn_mg11_T}
		\phi_{T^{(i)}}(s) = \frac{P_{\lambda-s}}{P_\lambda}.
	\end{equation}
        Note that the right hand side of \eqref{eqn:eqn_mg11_T} does not depend on the chosen stream.
\end{lemma}

\begin{proof}
	Without loss of generality we will prove Lemma~\ref{lemma:lemma_mg11_sys_time} for stream $1$. The system time $T_j$ of
	the $j^{th}$ successfully received packet corresponds to the service time of the $j^{th}$ received packet given that
	service was completed before any new arrival (since any new packet from any stream will preempt the current update
	being served). So, in steady-state, $\Prob\lp T>t\rp = \Prob\lp S>t|S<\min\lp X^{(1)},\dots, X^{(M)}\rp\rp$. Hence, for
	$L = \min\lp X^{(1)},\dots, X^{(M)}\rp$,
	\begin{align*}
		f_T(t) &= \lim_{\epsilon\to 0}\frac{\Prob\lp T\in [t,t+\epsilon] \rp}{\epsilon}\\
		       &= \lim_{\epsilon\to 0}\frac{\Prob\lp S\in [t,t+\epsilon]|S<L\rp}{\epsilon}\\
	      	       &= \lim_{\epsilon\to 0}\frac{\Prob\lp S\in [t,t+\epsilon]\rp \Prob\lp S<L| S\in
		       [t,t+\epsilon]\rp}{\epsilon\Prob\lp S<L\rp}\\
		       &= \frac{f_S(t)\Prob\lp L>t\rp}{\Prob\lp S<L\rp} = \frac{f_S(t)e^{-\lambda t}}{\Prob\lp S<L\rp},
	\end{align*}
	where the last equality is due to the fact that $L$ is exponentially distributed with rate $\lambda$. Thus,
	\begin{align*}
		\phi_T(s) &= \E\lp e^{sT} \rp = \int_0^\infty \frac{f_S(t)}{\Prob\lp
		S<L\rp}e^{-(\lambda-s)t}\mathrm{d}t = \frac{P_{\lambda-s}}{\Prob\lp S<L\rp}.
	\end{align*}
	Finally,
	\begin{align}
		\label{eq:eq_P_l}
		\Prob\lp S<L\rp &= \int_0^\infty f_S(t)\Prob\lp L>t\rp\mathrm{d}t = \int_0^\infty f_S(t)e^{-\lambda
		t}\mathrm{d}t\nn
				&= P_\lambda.
	\end{align}
\end{proof}

\begin{lemma}
	\label{lemma:lemma_mg11_Y}
	The moment generating function of the interdeparture time of the $i^{th}$ stream, $Y^{(i)}$, is 
	\begin{equation}
		\label{eqn:eqn_mg11_Y}
		\phi_{Y^{(i)}}(s) = \frac{\lambda_i P_{\lambda-s}}{\lambda_i P_{\lambda-s}-s}. 
	\end{equation}
\end{lemma}

\begin{figure}[!t]
\begin{tikzpicture}[>=stealth',shorten >=1pt,auto,node distance=2.8cm, semithick]
	\tikzstyle{every state}=[fill=red,draw=none,text=white]
	\node[initial,state] (A)                    {$q_0$};
	\node[state]         (B) [above right of=A] {$q_1$};
	\node[state]         (D) [below right of=A] {$q_{0'}$};
	\node[state]         (C) [below right of=B] {$q_{1'}$};
	
	\path [->] (A) edge [bend left]      node[ fill=white, anchor=center, pos=0.5] {$a$} (B);
	\path [->] (A) edge                  node[ fill=white, anchor=center, pos=0.5] {$z$} (C);
	\path [->] (B) edge [loop above]     node {$b$}(B);
	\path [->] (B) edge [bend left]      node[ fill=white, anchor=center, pos=0.5] {$u$} (A);
	\path [->] (B) edge [bend left]      node[ fill=white, anchor=center, pos=0.5] {$v$} (C);
	\path [->] (C) edge [bend left]      node[ fill=white, anchor=center, pos=0.5] {$u$} (D);
	\path [->] (C) edge [loop right]     node {$v$} (C);
	\path [->] (C) edge [bend left]      node[ fill=white, anchor=center, pos=0.5] {$b$} (B);
	\path [->] (D) edge [bend left]      node[ fill=white, anchor=center, pos=0.5] {$z$} (C);
	\draw [->] (D) ..controls +(east:5) and +(east:5)..(B) node[ fill=white, anchor=center, pos=0.5] {$a$};
\end{tikzpicture}
\caption{Semi-Markov chain representing the M/G/1/1 interdeparture time for stream $1$.}
\label{fig:fig_mg11_mc}
\end{figure}
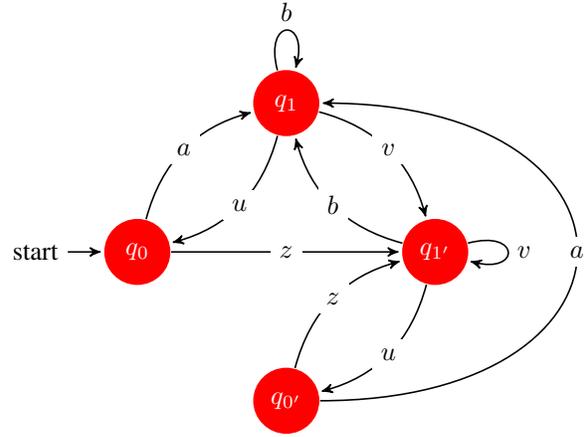

\begin{proof}
	Without loss of generality, we will prove Lemma~\ref{lemma:lemma_mg11_Y} for stream $1$. We define $L=\min\lp
	X^{(1)},\dots,X^{(M)}\rp$ and $\Lambda=\min\lp	X^{(2)},\dots,X^{(M)}\rp$. Since $L$ and $\Lambda$ are the minimum of
	independent exponential random variables, then they are also exponentially distributed with rates
	$\lambda=\lambda_1+\dots+\lambda_M$ and $\lambda-\lambda_1$ respectively. Fig.~\ref{fig:fig_mg11_mc} shows the semi-Markov chain
	relative to the interdeparture time $Y_j$ between the $j^{th}$ and ${j+1}^{th}$ received packet of the first stream.
	When the $j^{th}$ packet leaves the queue, the system enters the idle state $q_0$ where it waits for a new packet from
	any stream to be generated. Hence two clocks start: a clock $X^{(1)}$ and a clock $\Lambda$. Clock $X^{(1)}$ ticks first with 
	probability $a=\Prob\lp X^{(1)}<\Lambda\rp$, at which point a new packet from stream $1$ will be generated first and the system
	goes to state $q_1$. The value $A$ of the clock when it ticks has distribution $\Prob\lp A>t\rp=\Prob\lp
	X^{(1)}>t|X^{(1)}<\Lambda\rp$. Clock $\Lambda$ ticks first with probability $z=1-a=\Prob\lp \Lambda<X^{(1)}\rp$, at which point
	a new packet from one of the other $M-1$ streams is generated first and the system goes to state $q_{1'}$. The value
	$Z$ of this second clock when it ticks has distribution $\Prob\lp Z>t\rp=\Prob\lp \Lambda>t|\Lambda<X^{(1)}\rp$. 
	
	When the system arrives in state $q_1$, this means a packet from stream $1$ is starting service. Thus, due to the memoryless property 
	of $\Lambda$, three clocks start: a service clock $S$, clock $X^{(1)}$ and clock $\Lambda$. The service clock ticks
	first with probability $u=\Prob\lp S<L\rp$ and its value $U$ has distribution $\Prob\lp U>t\rp=\Prob\lp
	S>t|S<L\rp$. At this point the stream $1$ packet currently being served finishes service before any new packet is
	generated and the system goes back to state $q_0$. This ends the interdeparture time $Y_j$. On the other hand, clock
	$X^{(1)}$ ticks first with probability $b=\Prob\lp X^{(1)}<\min\lp S,\Lambda\rp\rp$ and its value $B$ has distribution
	$\Prob\lp B>t\rp=\Prob\lp X^{(1)}>t|X^{(1)}<\min\lp S,\Lambda\rp\rp$. At this point, a new stream $1$ update is
	generated before any other update from other streams and preempts the one currently in service. In this case the system
	stays in state $q_1$. The third clock $\Lambda$ ticks first with probability $v=\Prob\lp \Lambda<\min\lp
	S,X^{(1)}\rp\rp$ and its value $V$ has distribution $\Prob\lp V>t\rp=\Prob\lp \Lambda>t|\Lambda<\min\lp
	S,X^{(1)}\rp\rp$. At this point a new update not from stream $1$ is generated, preempts the one currently in service
	and the system switches to state $q_{1'}$. 

	When the system arrives in state $q_{1'}$, this means a packet not from stream $1$ is starting service. Thus, due to the memoryless property 
	of $X^{(1)}$, three clocks start: a service clock $S$, clock $X^{(1)}$ and clock $\Lambda$. As for state $q_1$, the service clock ticks
	first with probability $u$ and has value $U$. At this point packet currently being served finishes service before any new packet is
	generated and the system goes to state $q_{0'}$. Also like before, clock $X^{(1)}$ ticks first with probability $b$ and
	has value $B$. At this point, a new stream $1$ update is generated before any other update from other streams and preempts the
	one currently in service. In this case the system switches to state $q_1$. The third clock $\Lambda$ ticks first with
	probability $v$ and has value $V$. At this point a new update not from stream $1$ is generated, preempts the one currently in service
	and the system stays in state $q_{1'}$.

	Finally, when the system arrives in state $q_{0'}$, this means the system is idle but no update from stream $1$ has
	been delivered. Given $X^{(1)}$ and $\Lambda$ are memoryless, the system in state $q_{0'}$ behaves exactly like if it
	were in state $q_0$.

	From the above analysis we see that the interdeparture time is given by the sum of the values of the different clocks
	on the path starting and finishing at $q_0$. For example, for the path $q_0q_1q_{1'}q_{0'}q_{1'}q_1q_0$ in
	Fig.~\ref{fig:fig_mg11_mc} the interdeparture time $Y=A_1+V_1+U_1+Z_1+B_1+U_2$, where all the random variables in the sum are
	mutually independent. This value of $Y$ is also valid for the path $q_0q_{1'}q_{0'}q_1q_{1'}q_1q_0$. Hence $Y$ depends
	on the variables $A_j,B_j,U_j,V_j,Z_j$ and their number of occurrences and not on the path itself. Therefore, the
	probability that exactly $(i_1, i_2, i_3, i_4, i_5)$ occurrences of $\lp A, B, U, V, Z\rp$ happen, which is equivalent
	to the probability that $$Y=\sum_{k=1}^{i_1} A_k + \sum_{k=1}^{i_2} B_k + \sum_{k=1}^{i_3} U_k + \sum_{k=1}^{i_4} V_k +
	\sum_{k=1}^{i_5} Z_k$$ is given by $a^{i_1}b^{i_2}u^{i_3}v^{i_4}z^{i_5}Q(i_1,i_2,i_3,i_4,i_5)$, where 
	$Q(i_1,i_2,i_3,i_4,i_5)$ is the number of paths with this combination of occurrences. Taking into account the fact that
	the $\{A_k, B_k, U_k, V_k, Z_k\}$ are mutually independent, the moment generating function of $Y$ is
	\begin{align}
		\label{eq:eq_mg11_Y_mgf1}
		\phi_Y(s) &= \E\lp\E\lp e^{sY}|\lp I_1, I_2, I_3, I_4, I_5\rp=\lp i_1, i_2, i_3, i_4, i_5\rp
		\rp\rp\nn
			  &= \sum_{i_1,i_2,i_3,i_4,i_5}
			  \left[a^{i_1}b^{i_2}u^{i_3}v^{i_4}z^{i_5}Q(i_1,i_2,i_3,i_4,i_5)\right.\nn 
				  & \left. {} \E\lp e^{s\lp\sum_{k=1}^{i_1} A_k + \sum_{k=1}^{i_2} B_k + \sum_{k=1}^{i_3} U_k + \sum_{k=1}^{i_4} V_k +
			  \sum_{k=1}^{i_5} Z_k\rp}\rp\right]\nn
			  &=\sum_{i_1,i_2,i_3,i_4,i_5}
			  \left[a^{i_1}b^{i_2}u^{i_3}v^{i_4}z^{i_5}Q(i_1,i_2,i_3,i_4,i_5)\right.\nn
				  &\left. {} \E\lp e^{sA}\rp^{i_1}\E\lp e^{sB}\rp^{i_2}\E\lp e^{sU}\rp^{i_3}\E\lp
			  e^{sV}\rp^{i_4}\E\lp e^{sZ}\rp^{i_5}\right], 
	\end{align}
	where $\{I_1, I_2, I_3, I_4, I_5\}$ are the random variables associated with the number of occurrences of $\{A, B, U,
	V, Z\}$ respectively.
	
	Moreover, given a directed graph $G=(V,E)$ with algebraic label $L(e)$ on its
	edges, and a node $u\in V$ with no incoming edges, the transfer function
	$H(v)$ from $u$ to a node $v$ is the sum over of all paths from $u$ to
	$v$ with each path contributing the product of its edge labels to the
	sum (see \cite[pp. 213--216]{rimoldi_2016}). The complete set of transfer functions $\{H(v): v\in V\}$ can be
	computed easily by solving the linear equations:
	$$\begin{cases}
		H(u) & = 1\\
		H(w) &= \sum_{w': (w',w)\in E} H(w') L(( w',w)), \quad w\neq u.
	\end{cases}$$
	Observe that the sum in \eqref{eq:eq_mg11_Y_mgf1} is nothing but the transfer function from
	$q_0$ to $\bar q_0$ in the graph shown in Fig.~\ref{fig:fig_detour_flow_graph} with
	\begin{align*}
		&(D_1,D_2,D_3,D_4,D_5)\\
		&=\lp\E\lp e^{sA}\rp,\E\lp e^{sB}\rp,\E\lp e^{sU}\rp,\E\lp e^{sV}\rp,\E\lp e^{sZ}\rp\rp.
	\end{align*}
	Solving the system of linear equations above yields the transfer function as
	\begin{align}
		\label{eq:eq_H}
		&H(D_1,D_2,D_3,D_4,D_5) \nn
		&=\sum_{\substack{i_1,i_2,i_3,\\i_4,i_5}}
			  \left[Q(i_1,i_2,i_3,i_4,i_5)a^{i_1}b^{i_2}u^{i_3}v^{i_4}z^{i_5}\right.\nn
			  &\qquad \left. {} D_1^{i_1} D_2^{i_2}D_3^{i_3}D_4^{i_4}D_5^{i_5}\right]\nn
		&=\frac{uD_3\lp bD_2zD_5+aD_1-aD_1vD_4\rp}{\lp 1-bD_2\rp\lp1-uD_3zD_5\rp-vD_4\lp
		1+uD_3aD_1\rp}.
	\end{align}
	Thus $$\phi_Y(s) = H\lp\E\lp e^{sA}\rp,\E\lp e^{sB}\rp,\E\lp e^{sU}\rp,\E\lp e^{sV}\rp,\E\lp e^{sZ}\rp\rp.$$
	From Lemma~\ref{lemma:lemma_general_results}, we know that $\E\lp e^{sB}\rp=\E\lp
	e^{sV}\rp=\frac{\lambda\lp1-P_{\lambda-s}\rp}{\lp\lambda-s\rp\lp1-P_\lambda\rp}$ and $\E\lp e^{sA}\rp=\E\lp
	e^{sZ}\rp=\frac{\lambda}{\lambda-s}$. Moreover, one can notice that $U$ has the same distribution as the system time
	$T$ so $\E\lp e^{sU}\rp=\frac{P_{\lambda-s}}{P_\lambda}$. Simple computations show that 
	$a=\frac{\lambda_1}{\lambda}$, $b=\frac{\lambda_1}{\lambda}\lp1-P_\lambda \rp$, $u=P_\lambda$,
	$v=\frac{\lambda-\lambda_1}{\lambda}\lp1-P_\lambda\rp$, $z=\frac{\lambda-\lambda_1}{\lambda}$. Finally, replacing the above
	expressions into \eqref{eq:eq_H}, we get our result.
	
	\begin{figure}[!t]
	\begin{tikzpicture}[>=stealth',shorten >=1pt,auto,node distance=2.6cm, semithick]
		\tikzstyle{every state}=[fill=red,draw=none,text=white]
		\node[state] (A)                     {$q_0$};
		\node[state] (C) [right of=A]  	     {$q_{1'}$};
		\node[state] (B) [right of=C]  {$q_1$};
		\node[state] (E) [right of=B] 	     {$\bar q_{0}$};
		\node[state] (D) at (3.9,2.5)  {$q_{0'}$};
		
		\path [->] (A) edge [bend right]     node[ fill=white, anchor=center, pos=0.5] {$aD_1$} (B);
		\path [->] (A) edge                  node[ fill=white, anchor=center, pos=0.5] {$zD_5$} (C);
		\draw [->] (B) to [out=60, in=30, looseness=8] node[above] {$bD_2$} (B);
		\path [->] (B) edge                  node[ fill=white, anchor=center, pos=0.5] {$uD_3$} (E);
		\path [->] (B) edge [bend right=30]  node[ fill=white, anchor=center, pos=0.5] {$vD_4$} (C);
		\path [->] (C) edge [bend left=30]   node[ fill=white, anchor=center, pos=0.5] {$uD_3$} (D);
		\draw [->] (C) to [out=120, in=150, looseness=8] node[above] {$vD_4$}(C);
		\path [->] (C) edge       	     node[ fill=white, anchor=center, pos=0.5] {$bD_2$} (B);
		\path [->] (D) edge [bend left=20]   node[ fill=white, anchor=center, pos=0.5] {$zD_5$} (C);
		\draw [->] (D) edge                  node[ fill=white, anchor=center, pos=0.5] {$aD_1$} (B);
	\end{tikzpicture}
	\caption{Detour flow graph of the M/G/1/1 interdeparture time for stream $1$.}
	\label{fig:fig_detour_flow_graph}
	\end{figure}
\end{proof}
\begin{thm}
	\label{thm:thm_mg11_age}
	Given an M/G/1/1 queue with preemption and service time $S$ and a source generating packets belonging to $M$ streams according to $M$
	independent Poisson processes with rates $\lambda_i$, $i=1,\dots,M$, such that $\lambda=\lambda_1+\dots+\lambda_M$,
	then  
	\begin{enumerate}
		\item the average age of stream $i$ is given by
			\begin{equation}
				\label{eq:eq_avg_age_mg11}
				\Delta_i = \frac{1}{\lambda_i P_\lambda},
			\end{equation}
		\item and the average peak age of stream $i$ is given by
			\begin{equation}
				\label{eq:eq_avg_peak_age_mg11}
				\Delta_{peak,i} = \frac{1}{\lambda_i P_\lambda}+\frac{\E\left(Se^{-\lambda S}\right)}{P_\lambda}.
			\end{equation}
	\end{enumerate}
\end{thm}
\begin{proof}
	Due to the symmetry in the system from a stream point of view, then, without loss of generality, we will prove 
	\ref{thm:thm_mg11_age} for stream $1$ only. The same proof applies for the other $M-1$ streams.
	
	From \eqref{eq:eq_age_definition} and Fig.~\ref{fig1}, the average age for stream $1$ of the M/G/1/1 queue can
	be also expressed as the sum of the geometric areas $Q_i$ under the instantaneous age curve. Authors in 
	\cite{2012CISS-KaulYatesGruteser} show that
	\begin{align}
		\label{eq:eq_agedef}
		\Delta_1 &=\lim_{\tau\to\infty}\frac{N_\tau}{\tau}\frac{1}{N_\tau}\sum_{j=1}^{N_\tau} Q_j
		= \lambda_e\mathbb{E}(Q),
	\end{align}
	where $\lambda_e =  \lim_{\tau\to\infty}\frac{N_\tau}{\tau}$, $Q$ is the steady-state counterpart of $Q_j$ and the
	second equality is justified by the ergodicity of the system. As shown in \cite{YatesKaul-2016arxiv} and 
	\cite{NajmYatesSoljaninMG11}, $\lambda_e$ is the rate at which successful updates are received. Given that the 
	interarrival time of all streams are memoryless, then the interdeparture times,	$Y_j$ and $Y_{j+1}$, between two 
	consecutive received updates are i.i.d. Hence $N_{\tau}$ forms a renewal process and by \cite{ross},
	$\lim_{\tau\to\infty}\frac{N_\tau}{\tau}=\frac{1}{\E(Y)}$, where $Y$ is the steady-state interdeparture random
	variable. Moreover, from Fig.~\ref{fig1} we see that by applying same reasoning as in \cite{NajmNasser-ISIT2016} 
	$$ \E\left(Q\right) = \frac{1}{2} \E\lp Y^2\rp +\E\lp TY\rp = \frac{1}{2} \E\lp Y^2\rp +\E\lp T\rp \E\lp
	Y\rp.$$ The second equality is obtained by noticing that for any received packet $j$, $T_j$ and $Y_j$ are independent.
	Therefore,
	\begin{equation}
		\label{eq:eq_steady_state_mg11_age}
		\Delta_1 = \E\lp T\rp + \frac{\E\lp Y^2\rp}{2\E\lp Y\rp}
	\end{equation}
	Moreover, from Fig.~\ref{fig1} we see that the peak age at the instant before receiving the $j^{th}$ packet is given
	by $$K_j = T_{j-1}+Y_{j-1}.$$ Hence, given that the system is ergodic, \eqref{eq:eq_avg_peak_age_def} becomes at steady state, 
	\begin{equation}
		\label{eq:eq_peak_age_1}
		\Delta_{peak,1}=\E\lp K\rp=\E\lp T\rp+\E\lp Y\rp.
	\end{equation}
	
	Using Lemma~\ref{lemma:lemma_mg11_sys_time}, we get that $\E\lp T\rp= P_\lambda^{-1}\E\lp Se^{-\lambda S}\rp$. Using
	Lemma~\ref{lemma:lemma_mg11_Y}, we get that $\E\lp Y\rp= \lp \lambda_1 P_\lambda\rp^{-1}$ and $\E\lp
	Y^2\rp=2\lp-\frac{\E\lp Se^{-\lambda S}\rp}{\lambda_1 P_\lambda^2}+\frac{1}{\lambda_1^2 P_\lambda^2}\rp$. Using these 
	expressions in \eqref{eq:eq_steady_state_mg11_age} and \eqref{eq:eq_peak_age_1} we obtain our result for stream $1$. 
\end{proof}

Note that for $M=1$, we get back the result derived in \cite{NajmYatesSoljaninMG11} for single stream M/G/1/1 preemptive queue.
Moreover, if we replace $P_\lambda$ in \eqref{eq:eq_avg_age_mg11} by the Laplace transform of the exponential
distribution evaluated at $\lambda$ we recover the expression stated in \cite[Theorem 2(a)]{YatesKaul-2016arxiv}.

\begin{corollary}
	\label{cor:cor}
	Let a source generate updates according to a Poisson process with fixed rate $\lambda$. Moreover, these updates belong
	to $M$ different streams, each stream $i$ chosen independently with probability $p_i$ at generation time. Then if we
	use an M/G/1/1 with preemption transmission scheme we can decrease the average age (and the average peak age) of a high
	priority stream $k$ with respect to the other streams by increasing the probability $p_k$ with which it is chosen.
\end{corollary}
\begin{proof}
	From Theorem~\ref{thm:thm_mg11_age} we know that for any two streams $i$ and $k$, in order to have $\Delta_i<\Delta_k$
	or $\Delta_{peak,i}<\Delta_{peak,k}$ we must have $\lambda_i>\lambda_k$. Given that $\lambda_i = \lambda p_i$, 
	$i=1,\dots,M$, then  we must have $p_i>p_k$.
\end{proof}
Given the source generates multiple streams, one interesting performance measure of the system would be the total average age
or total average peak age defined respectively as 
\begin{align}
	\label{eq:eq_total_age}
	\Delta_{tot} &= \sum_{i=1}^M \Delta_i,\ \Delta_{peak,tot} = \sum_{i=1}^M \Delta_{peak,i}.
\end{align}
The next theorem gives the distribution over the $p_i$, $i=1,\dots,M$, that minimizes the metrics in \eqref{eq:eq_total_age} as
well as their minimum achievable value.
\begin{thm}
	\label{thm:thm_optimal_age}
	For the M/G/1/1 multi-stream preemptive system described above with fixed total generation rate $\lambda$, the optimal
	strategy that achieves the smallest value for the total average age, $\Delta_{tot}$, and the total average peak age,
	$\Delta_{peak,tot}$, is the fair strategy: all streams should have the same generation rate. This means that the
	probability distribution $\{p_i\}$ over the choices of streams should be the uniform distribution with
	$p_i=\frac{1}{M}$, $i=1,\dots,M$. Moreover, the optimal values of $\Delta_{tot}$ and $\Delta_{peak,tot}$ are given by
	\begin{align}
		\label{eq:eq_optimal_total_age}
		\Delta_{tot} &= \frac{M^2}{\lambda P_\lambda},\ \Delta_{peak,tot} = \frac{M^2}{\lambda P_\lambda}+\frac{M\E\left(Se^{-\lambda S}\right)}{P_\lambda}
	\end{align}
\end{thm}

\begin{proof}
	From \eqref{eq:eq_avg_age_mg11}, \eqref{eq:eq_avg_peak_age_mg11} and \eqref{eq:eq_total_age}, we get that 
	\begin{align}
		\label{eq:eq_age_tot_2}
		\Delta_{tot} &= \frac{1}{P_\lambda}\sum_{i=1}^M \frac{1}{\lambda_i} = \frac{1}{\lambda P_{\lambda}}\sum_{i=1}^M
		\frac{1}{p_i}\nn
		\Delta_{peak,tot} &= \frac{1}{P_\lambda}\sum_{i=1}^M \frac{1}{\lambda_i}+\frac{M\E\left(Se^{-\lambda
		S}\right)}{P_\lambda}\nn
				  &= \frac{1}{\lambda P_\lambda}\sum_{i=1}^M \frac{1}{p_i}+\frac{M\E\left(Se^{-\lambda
		S}\right)}{P_\lambda}
	\end{align}
	Given that $\lambda$ is fixed, then minimizing $\Delta_{tot}$ and $\Delta_{peak,tot}$ over $(p_1,\dots,p_M)$ is
	equivalent to minimizing $\sum_{i=1}^M \frac{1}{p_i}$. As this is a symmetric convex function, it is minimized when
	$p_1=\dots=p_M=1/M$ with the value $M^2$, which proves our theorem.
	\end{proof}
	From Corollary~\ref{cor:cor} and Theorem~\ref{thm:thm_optimal_age}, we see that prioritizing a stream over the others
	from an age point of view and minimizing the total age are two contradictory objectives. 


\section{Conclusion}

In this paper we studied the M/G/1/1 preemptive system with a multi-stream updates source. We derived closed form expressions
for the average age and average peak age using the detour flow graph method. Moreover, using these results we showed that, for
a fixed total generation rate, one can't prioritize one of the streams and at the same time minimize the total age. In fact, we
prove that in order to optimize the total age, the source needs to generate all streams at the same rate. This means that no
single stream can be given a higher rate, a necessary condition to reduce its age with respect to the other streams.


\section*{Acknowledgements}

This research was supported in part by grant No.  200021\_166106/1 of the Swiss National Science Foundation.



\bibliographystyle{IEEEtran}
\bibliography{IEEEabrv,ehcache}
%



\end{document}